\def\UseBibLatex{1}

\documentclass[12pt]{article}%

\usepackage{amsmath}%
\usepackage{graphicx}%
\graphicspath{{./figs/}}
\usepackage[cm]{fullpage}%
\usepackage{amssymb}%
\usepackage{xcolor}%
\usepackage{mleftright}%
\usepackage{scalerel}%
\usepackage{xspace}%
\usepackage{caption}%
\usepackage{multirow}

\usepackage{titlesec}%
\titlelabel{\thetitle. }%

\usepackage{hyperref}%

\providecommand{\IfPrinterVer}[2]{#2}%
\IfPrinterVer{%
   \usepackage{hyperref}%
}{%
   \usepackage{hyperref}%
   \hypersetup{%
      breaklinks,%
      colorlinks=true,%
      urlcolor=[rgb]{0.25,0.0,0.0},%
      linkcolor=[rgb]{0.5,0.0,0.0},%
      citecolor=[rgb]{0,0.2,0.445},%
      filecolor=[rgb]{0,0,0.4},
      anchorcolor=[rgb]={0.0,0.1,0.2}%
   }
}

\numberwithin{figure}{section}%
\numberwithin{table}{section}%
\numberwithin{equation}{section}%

\IfFileExists{sariel_computer.sty}{\def\sarielComp{1}}{}
\ifx\sarielComp\undefined%
\newcommand{\SarielComp}[1]{}
\newcommand{\NotSarielComp}[1]{#1}%
\else
\newcommand{\SarielComp}[1]{#1}%
\newcommand{\NotSarielComp}[1]{}%
\fi
\providecommand{\BibLatexMode}[1]{}
\providecommand{\BibTexMode}[1]{#1}

\ifx\UseBibLatex\undefined%
  \renewcommand{\BibLatexMode}[1]{}
  \renewcommand{\BibTexMode}[1]{#1}
\else
  \renewcommand{\BibLatexMode}[1]{#1}
  \renewcommand{\BibTexMode}[1]{}
\fi

\BibLatexMode{%
   \usepackage[bibencoding=utf8,style=alphabetic,backend=biber]{biblatex}%
   \usepackage{sariel_biblatex}%
}

\newcommand{\Of}{\Mh{\mathcal{O}}}%

\newcommand{\BSet}{\Mh{{B}}}%
\newcommand{\EBSet}{\Mh{{B^+}}}%

\usepackage[amsmath,thmmarks]{ntheorem}%
\theoremseparator{.}%

\theoremstyle{plain}%
\newtheorem{theorem}{Theorem}[section]
\theoremsymbol{}%

\newtheorem{lemma}[theorem]{Lemma}

\theoremstyle{plain}%
\theoremheaderfont{\sf} \theorembodyfont{\upshape}%
\newtheorem*{remark:unnumbered}[FakeCounter]{Remark}%
\newtheorem{remark}[theorem]{Remark}%
\newtheorem{definition}[theorem]{Definition}
\newtheorem*{defn:unnumbered}[FakeCounter]{Definition}

\newcommand{\myqedsymbol}{\rule{2mm}{2mm}}

\theoremheaderfont{\em}%
\theorembodyfont{\upshape}%
\theoremstyle{nonumberplain}%
\theoremseparator{}%
\theoremsymbol{\myqedsymbol}%
\newtheorem{proof}{Proof:}%

\newcommand{\eps}{\varepsilon}%

\newcommand{\epsB}{\varsigma}

\newcommand{\epsR}{\Mh{\vartheta}}%
\newcommand{\epsP}{\Mh{\rho}}%

\newcommand{\ceil}[1]{\left\lceil {#1} \right\rceil}
\newcommand{\floor}[1]{\left\lfloor {#1} \right\rfloor}

\newcommand{\HLinkShort}[2]{\hyperref[#2]{#1\ref*{#2}}}
\newcommand{\HLink}[2]{\hyperref[#2]{#1~\ref*{#2}}}
\newcommand{\HLinkPage}[2]{\hyperref[#2]{#1~\ref*{#2}%
      $_\text{p\pageref{#2}}$}}
\newcommand{\HLinkPageOnly}[1]{\hyperref[#1]{Page~\refpage*{#1}%
      $_\text{p\pageref{#1}}$}}

\newcommand{\HLinkSuffix}[3]{\hyperref[#2]{#1\ref*{#2}{#3}}}
\newcommand{\HLinkPageSuffix}[3]{\hyperref[#2]{#1\ref*{#2}%
      #3$_\text{p\pageref{#2}}$}}

\newcommand{\figlab}[1]{\label{fig:#1}}
\newcommand{\figref}[1]{\HLink{Figure}{fig:#1}}

\newcommand{\seclab}[1]{\label{sec:#1}}
\newcommand{\secref}[1]{\HLink{Section}{sec:#1}}

\providecommand{\deflab}[1]{\label{def:#1}}
\newcommand{\defref}[1]{\HLink{Definition}{def:#1}}

\newcommand{\defrefY}[2]{\hyperref[def:#2]{#1}}

\newcommand{\lemlab}[1]{\label{lemma:#1}}
\newcommand{\lemref}[1]{\HLink{Lemma}{lemma:#1}}%

\newcommand{\tablab}[1]{\label{table:#1}}%
\newcommand{\tabref}[1]{\HLink{Table}{table:#1}}%

\newcommand{\thmlab}[1]{{\label{theo:#1}}}
\newcommand{\thmref}[1]{\HLink{Theorem}{theo:#1}}

\providecommand{\eqlab}[1]{}%
\renewcommand{\eqlab}[1]{\label{equation:#1}}
\newcommand{\Eqref}[1]{\HLinkSuffix{Eq.~(}{equation:#1}{)}}

\providecommand{\Mh}[1]{{#1}}%

\IfFileExists{.latex_printer_friendly}{\def\GenPrinterVer{1}}{}%

\ifx\GenPrinterVer\undefined
   \IfFileExists{.latex_color}{\def\GenColorMath{1}}{}
\else
   \renewcommand{\IfPrinterVer}[2]{#1}%
\fi

\ifx\GenColorMath\undefined
\else
\renewcommand{\Mh}[1]{{\textcolor{red}{#1}}}%
\fi

\definecolor{almostblack}{rgb}{0, 0, 0.3}
\providecommand{\emphw}[1]{{\textcolor{almostblack}{\emph{#1}}}}%
\definecolor{blue25emph}{rgb}{0, 0, 11}
\newcommand{\emphic}[2]{%
   \textcolor{blue25emph}{%
      \textbf{\emph{#1}}}%
   \index{#2}}
\newcommand{\emphi}[1]{\emphic{#1}{#1}}%
\providecommand{\si}[1]{#1}

\SarielComp{%
\definecolor{blue25emph}{rgb}{0, 0, 11}
\renewcommand{\emphic}[2]{%
   \textcolor{blue25emph}{%
      \textbf{\emph{#1}}}%
   \index{#2}}

}

\usepackage{stmaryrd}%
\newcommand{\IntRange}[1]{\mleft[ #1 \mright]}
\newcommand{\IRX}[1]{\IntRange{#1}}%
\newcommand{\IRY}[2]{\left[ #1 \ldots #2 \right]}

\renewcommand{\Re}{\mathbb{R}}%
\newcommand{\NN}{\mathbb{N}}%

\newcommand{\PS}{\Mh{P}}%

\newcommand{\QS}{\Mh{Q}}%

\newcommand{\Graph}{\Mh{G}}%
\newcommand{\VV}{\Mh{V}}%
\newcommand{\VX}[1]{\VV\pth{#1}}%
\newcommand{\Edges}{\Mh{E}}%
\newcommand{\EdgesX}[1]{\Mh{E}\pth{#1}}%

\newcommand{\Set}[2]{\left\{ #1 \;\middle\vert\; #2 \right\}}
\newcommand{\cardin}[1]{\left| {#1} \right|}%
\newcommand{\pth}[2][\!]{\mleft({#2}\mright)}%
\newcommand{\brc}[1]{\left\{ {#1} \right\}}

\newcommand{\dGY}[2]{\Mh{\mathsf{d}}\pth{#1,#2}}%
\newcommand{\dGZ}[3]{\Mh{\mathsf{d}_{#1}}\pth{#2,#3}}%
\newcommand{\dY}[2]{\left\| #1 - #2 \right\|}%

\newcommand{\pp}{\Mh{p}}%
\newcommand{\pq}{\Mh{q}}%
\newcommand{\pz}{\Mh{z}}%
\newcommand{\etal}{\textit{et~al.}\xspace}
\newcommand{\ShadowC}{\Mh{\mathcal{S}}}%

\newcommand{\ShadowY}[2]{\Mh{\ShadowC}\pth{#1,#2}}%

\newcommand{\ShadowLY}[2]{\Mh{\ShadowC_\rightarrow}\pth{#1,#2}}%
\newcommand{\ShadowRY}[2]{\Mh{\ShadowC_\leftarrow}\pth{#1,#2}}%

\providecommand{\qedhere}{}

\renewcommand{\th}{th\xspace}

\newcommand{\pbrcx}[1]{\left[ {#1} \right]}%
\newcommand{\ProbLTR}{\mathbb{P}}%

\newcommand{\Prob}[1]{\mathop{\ProbLTR} \mleft[ #1 \mright]}%
\newcommand{\ProbCond}[2]{\mathop{\ProbLTR}\!\left[%
       #1 \;\middle\vert\; #2 \right]}

\newcommand{\ExChar}{\mathbb{E}}%
\newcommand{\ExSym}{\mathop{\ExChar}}%
\newcommand{\Ex}[2][\!]{\ExSym#1\pbrcx{#2}}

\newcommand{\ExCond}[2]{\ExSym\!\left[%
       #1 \;\middle\vert\; #2 \right]}

\usepackage[inline]{enumitem}

\newlist{compactenumA}{enumerate}{5}%
\setlist[compactenumA]{topsep=0pt,itemsep=-1ex,partopsep=1ex,parsep=1ex,%
  label=(\Alph*)}%

\newlist{compactenuma}{enumerate}{5}%
\setlist[compactenuma]{topsep=0pt,itemsep=-1ex,partopsep=1ex,parsep=1ex,%
  label=(\alph*)}%

\newlist{compactenumI}{enumerate}{5}%
\setlist[compactenumI]{topsep=0pt,itemsep=-1ex,partopsep=1ex,parsep=1ex,%
   label=(\Roman*)}%

\newlist{compactenumi}{enumerate}{5}%
\setlist[compactenumi]{topsep=0pt,itemsep=-1ex,partopsep=1ex,parsep=1ex,%
  label=(\roman*)}%

\newlist{compactitem}{itemize}{5}%
\setlist[compactitem]{topsep=0pt,itemsep=-1ex,partopsep=1ex,parsep=1ex,\label=ensuremath{\bullet}}%

\newlength{\savedparindent}

\providecommand{\remove}[1]{}%

\newcommand{\Here}{\typeout{LOCATION: \currfilename\space L\the\inputlineno}\xspace}

\newcommand{\Daniel}{D\'aniel\xspace}%
\newcommand{\Olah}{Ol\'ah\xspace}%

\providecommand{\Mh}[1]{{#1}}%

\newcommand{\ballY}[2]{\mathrm{ball}\pth{#1, #2}}%

\newlength{\ppicX}
\newlength{\ppicY}

\newcommand{\order}{\sigma}
\newcommand{\orderset}{\Pi}
\newcommand{\ordAll}{\orderset^+}%

\newcommand{\Event}{\Mh{\mathcal{E}}}%

\newcommand{\atgen}{\symbol{'100}}
\newcommand{\SarielThanks}[1]{\thanks{Department of Computer Science;
      University of Illinois; 201 N. Goodwin Avenue; Urbana, IL,
      61801, USA; {\tt sariel\atgen{}illinois.edu}; {\tt
         \url{http://sarielhp.org/}.} #1}}

\newcommand{\KevinThanks}[1]{%
   \thanks{%
      Department of Mathematics and Computing Science, TU Eindhoven,
      P.O. Box 513, 5600 MB Eindhoven, The Netherlands. %
      #1}%
}
\newcommand{\OlahThanks}[1]{%
   \thanks{%
      Department of Mathematics and Computing Science, TU Eindhoven,
      P.O. Box 513, 5600 MB Eindhoven, The Netherlands. %
      #1}%
}

\newcommand{\SaveContent}[2]{%
   \expandafter\newcommand{#1}{#2}%
}

\newcommand{\pa}{\Mh{p}}%
\newcommand{\pb}{\Mh{q}}%

\newcommand{\IV}{\Mh{J}}%

\newcommand{\UpY}[2]{#1^{}_{\uparrow {#2}}}

\newcommand{\itop}{\Mh{\xi}}

\newcommand{\SH}{\Mh{\mathcal{S}}}%

\newcommand{\constA}{\Mh{c}}%

\newcommand{\connX}[1]{\Mh{\ell}\pth{#1}}%
\renewcommand{\sp}{\Mh{\alpha}}
\newcommand{\lossY}[2]{\Mh{\lambda}\pth{#1,#2}}

\newcommand{\StwSet}{\Mh{{T}}}%

\newcommand{\compX}[1]{#1^c}
\allowdisplaybreaks[4]

\BibLatexMode{%
   \bibliography{qreliable}
}
\begin{document}

\title{Sometimes Reliable Spanners of Almost Linear Size}%

\author{%
   Kevin Buchin%
   \KevinThanks{}%
   \and%
   Sariel Har-Peled%
   \SarielThanks{Work on this paper was partially supported by a NSF
      AF awards CCF-1421231 and CCF-1907400.}%
   \and%
   \Daniel \Olah%
   \OlahThanks{Supported by the Netherlands Organisation for
      Scientific Research (NWO) through Gravitation-grant
      NETWORKS-024.002.003.}%
}%

\maketitle

\begin{abstract}
    Reliable (Euclidean) spanners can withstand huge failures, even
    when a linear number of vertices are deleted from the network. In
    case of failures, some of the remaining vertices of a reliable
    spanner may no longer admit the spanner property, but this
    collateral damage is bounded by a fraction of the size of the
    attack. It is known that $\Omega(n\log n)$ edges are needed to
    achieve this strong property, where $n$ is the number of vertices
    in the network, even in one dimension. Constructions of reliable
    geometric~$(1+\eps)$-spanners, for $n$ points in $\Re^d$, are
    known, where the resulting graph has
    $\Of( n \log n \log\!\log^{6}\!n )$ edges.

    Here, we show randomized constructions of smaller size Euclidean
    spanners that have the desired reliability property in expectation
    or with good probability. The new construction is simple, and
    potentially practical -- replacing a hierarchical usage of
    expanders (which renders the previous constructions impractical)
    by a simple skip list like construction.  This results in a
    $1$-spanner, on the line, that has linear number of edges. Using
    this, we present a construction of a reliable spanner in $\Re^d$
    with~$\Of( n \log\!\log^{2}\!n \log\!\log\!\log n )$ edges.
\end{abstract}

\section{Introduction}

\paragraph{Motivation.}

Designing efficient networks (of roads, railways, communications, etc)
is an important problem. Such networks have to be efficient -- that
is, the distance of traveling along them between two points should be
close to their real world distance. The design and study of such
constructions is the task of building good \emph{spanners}. From an
adversarial point of view, a natural task is figuring out what nodes
to take out in order to ``break'' the network -- so that the
connectivity of the network is degraded. For example, this problem was
studied during the cold war in choosing targets to attack by nuclear
weapons \cite{k-bpgsh-20}:
\begin{quote}
    \scalebox{1.4}{\textbf{``}}%
    {\em As applied to warfare, nodal analysis examined a system or
       network with an eye toward identifying the central links—the
       nodes—so that the commanders could allocate their weapons most
       cost-effectively: destroy the nodes, which in some cases would
       take only a few weapons, and the whole system falls
       apart. There were some officers in the SAC Air Room who
       specialized in this sort of analysis...}%
    \scalebox{1.4}{\textbf{''}}
\end{quote}
Here, we study the natural ``dual'' problem of constructing networks
that can withstand widespread attacks/node-failure.

\paragraph{Background.}
Geometric graphs are such that their vertices are points in the
$d$-dimensional Euclidean space $\Re^d$ and edges are straight line
segments. The quality or efficiency of a geometric graph is often
measured in terms of the ratio of shortest path distances and
geometric distances between its vertices. Let $\Graph=(\PS,\Edges)$ be
a geometric graph, where $\PS \subset \Re^d$ is a set of $n$ points
and $\Edges$ is the set of edges. The shortest path distance between
two points $\pa,\pb \in \PS$ in the graph $\Graph$ is denoted by
$\dGZ{\Graph}{p}{q}$ (or just $\dGY{p}{q}$). The graph $\Graph$ is a
$t$-spanner for some constant $t\geq 1$, if
$\dGY{p}{q} \leq t \cdot \dY{p}{q}$ holds for all pairs of points
$\pa,\pb \in \PS$, where $\dY{p}{q}$ stands for the Euclidean distance
of $\pa$ and $\pb$. The spanning ratio, stretch factor, or dilation of
a graph $\Graph$ is the minimum number $t\geq 1$ for which $\Graph$ is
a $t$-spanner. A path between $\pa$ and $\pb$ is a \emph{$t$-path} if
its length is at most $t\cdot \dY{p}{q}$.

Here, we construct spanners that can survive massive failures of
vertices. The most studied notion is fault tolerance
\cite{clns-nds-15,lns-eacft-98,lns-iacft-02,l-nrftg-99,s-ofts-14},
which provides a properly functioning residual graph if there are no
more failures than a predefined parameter $k$.  Formally, a
$t$-spanner is \emphw{$k$-fault tolerant} if the graph remains a
$t$-spanner after any $k$ vertices are deleted.  It is clear, that a
$k$-fault tolerant spanner must have~$\Omega\pth{kn}$ edges to avoid
small degree nodes, which can be isolated by deleting their
neighbors. Therefore, fault tolerant spanners must have quadratic size
to be able to survive a failure of a constant fraction of vertices.

Another notion is robustness \cite{bdms-rgs-13}, which gives more
flexibility by allowing the loss of some additional nodes by not
guaranteeing $t$-paths for them. For a
function~$f:\NN \xrightarrow{} \Re^+$ a $t$-spanner $\Graph$ is
$f$-robust, if for any set of failed points $\BSet$ there is an
extended set $\EBSet$ (that contains $\BSet$) with size at most
$f\pth{\cardin{\BSet}}$ such that the residual graph
$\Graph\setminus\BSet$ has a $t$-path for any pair of
points~$\pa,\pb \in \PS \setminus \EBSet$. The function $f$ controls
the robustness of the graph - the slower the function grows the more
robust the graph is.  For~$\epsR\in(0,1)$, a spanner that is
$f$-robust with~$f(k)=(1+\epsR)k$ is a \emph{$\epsR$-reliable}
spanner~\cite{bho-spda-20}.  This is the strongest form of robustness,
since the dilation can increase beyond $t$ only for a tiny additional
fraction of points. The fraction is relative to the number of failed
vertices and controlled by the parameter $\epsR$.

For the sake of simplicity of exposition, we omit the $\epsR$ and $t$
parameters when they are understood from the context. Thus, in the
following, a \emph{reliable spanner} is a shorthand for a
$\epsR$-reliable $t$-spanner.

Recently, Buchin \etal \cite{bho-spda-20} showed a construction of
reliable $1$-spanners of size $\Of\pth{n\log n}$ in one
dimension, and of reliable $(1+\eps)$-spanners of size
$\Of\pth{n\log n \log\!\log^6\!n }$ in higher dimensions (the constant
in the $\Of$ depends on the dimension, $\eps$, and the reliability
parameter). An alternative construction, with slightly worse bounds,
was given by Bose~\etal~\cite{bcdm-norgm-18}.  Up to polynomial
factors in $\log \log n$, this matches a lower bound of Bose \etal
\cite{bdms-rgs-13}.

\paragraph{Limitations of previous constructions.} %
The construction of Buchin \etal \cite{bho-spda-20} (and also the
construction of Bose \etal \cite{bcdm-norgm-18}) relies on using
expanders to get a $1$-spanner for points on the line, and then
extending it to higher dimensions. The spanner (in one dimension) has
$\Of(n\log n)$ edges. Unfortunately, even in one dimension, such a
reliable spanner requires~$\Omega( n \log n)$ edges, as shown by Bose
\etal \cite{bdms-rgs-13}. Furthermore, the constants involved in these
constructions \cite{bcdm-norgm-18,bho-spda-20} are quite bad, because
of the usage of expanders.  See \tabref{size-summary} for a summary of
the sizes of different constructions (together with the new results).

\begin{table}
    \begin{center}
        \begin{tabular}{*{5}{|c}|}
          \hline
          \hline
          & dim
          & \# edges
          & constants \\%
          \hline
          \hline
          \multicolumn{4}{|l|}{Reliable spanners}
          \\
          \hline
          \hline
          \multirow{2}{*}{Buchin \etal \cite{bho-spda-20}}
          & $d=1$
          & $\Of(n \log n)$
          & $\Of\pth{\epsR^{-6}}$
          \\
          & $\Bigl. d\geq 2$
          & $\Of\pth{n \log n \log\!\log^6\! n}$
          & $\Of\pth{\eps^{-7d} \epsR^{-6} \log^7\!\eps^{-1}}$
          \\
          \hline%
          Bose \etal \cite{bcdm-norgm-18}
          & $\Bigl. d\geq 1$
          & $\Of\pth{n \log^2\! n \log\!\log n}$
          & ?\\
          \hline%
          \hline%
          \multicolumn{4}{|l|}{Reliable spanners \emph{in expectation}}
          \\
          \hline
          \hline
          \multirow{2}{*}{New results}
          & $\Bigl. d=1$
          & $\Of(n)$
          & $\Of\pth{\epsR^{-1} \log \epsR^{-1}}$
          \\
          & $\Bigl. d\geq 2$
          & $\Of(n \log\!\log^2\! n \log\!\log\!\log n )$
          & $\Of\pth{\eps^{-2d} \epsR^{-1} \log^3\! \eps^{-1} \log \epsR^{-1}}$
          \\
          \hline
          \hline
        \end{tabular}%
    \end{center}
    \caption{Comparison of the size of constructions of reliable spanners and reliable spanners in expectation. The reliability parameter is $\epsR>0$, and, for dimensions $d\geq 2$, the graphs are $(1+\eps)$-spanners for $\eps > 0$.}
    \tablab{size-summary}
\end{table}

\paragraph{The problem.} %
As such, the question is whether one can come up with simple and
practical constructions of spanners that have linear or near linear
size, while still possessing some reliability guarantee -- either in
expectation or with good probability.

\paragraph{Some definitions.}
Given a graph $\Graph$, an \emph{attack} $\BSet \subseteq \VX{\Graph}$
is a non-empty set of vertices that are being removed. Informally, a
\emph{damaged set} $\EBSet \supseteq \BSet$, is a set of vertices
which are no longer connected to the rest the graph, or are badly
connected to the rest of the graph -- that is, the graph
$\Graph \setminus \BSet$ is a good spanner for all the points of
$\PS \setminus \EBSet$.  While $\EBSet$ is not unique, we try to pick
a damaged set that is as small as possible.  The \emph{loss} caused by
$\BSet$, is the quantity $\cardin{\EBSet \setminus \BSet}$ (where we
take the minimal damaged set). The \emph{loss rate} of $\BSet$ is
$\lossY{\Graph}{\BSet} = \cardin{\EBSet \setminus \BSet} /
\cardin{\BSet}$. A graph $\Graph$ is \emph{$\epsR$-reliable} if for
any attack $\BSet$, the loss rate $\lossY{\Graph}{\BSet}$ is at most
$\epsR$.  See \defref{formal:r:spanner} for the formal definition.

\paragraph{Randomness and obliviousness.}
As mentioned above, reliable spanners must have size
$\Omega( n \log n)$. A natural way to get a smaller spanner, is to
consider randomized constructions, and require that the reliability
holds in expectation (or with good probability). Randomized
constructions are (usually) still sensitive to adversarial attacks, if
the adversary is allowed to pick the attack set after the construction
is completed (and it is allowed to inspect it). A natural way to deal
with this issue is to restrict the attacks to be \emph{oblivious} --
that is, the attack set is chosen before the graph is constructed (or
without any knowledge of the edges).

In such an oblivious model, the loss rate is a random variable (for a
fixed attack $\BSet$). It is thus natural to construct the graph
$\Graph$ randomly, in such a way that
$\Ex{\lossY{\Graph}{\BSet}} \leq \epsR$, or alternatively, that the
probability $\Prob{\lossY{\Graph}{\BSet} \geq \epsR }$ is small.

\paragraph{$1$-spanner.}
Surprisingly, the one-dimensional problem is the key for building
reliable spanners. Here, the graph $\Graph$ is constructed over the
set of vertices $\IRX{n} = \{1,\ldots, n\}$. An attack is a subset
$\BSet \subseteq \IRX{n}$.  Given an attack $\BSet$, the requirement
is that for all $i,j \in \IRX{n} \setminus \EBSet$, such that $i <j$,
there is a monotonically increasing path\footnote{A path
   $i_1 i_2 \ldots i_k$ is \emph{monotonically increasing} if
   $i_1 < i_2 < \cdots < i_k$.} from $i$ to $j$ in
$\Graph \setminus \BSet$ -- here, the length of the path between $i$
and $j$ is exactly $j-i$. Since there is no distortion in the length
of the path, such graphs are $1$-spanners.

\paragraph{Our results.} %
We give a randomized construction of a $1$-spanner in one dimension,
that is $\epsR$-reliable in expectation, and has size
$\Of\pth{n}$. Formally, the construction has the property that
$\Ex{\lossY{\Graph}{\BSet}} \leq \epsR$.  This construction can also
be modified so that $\lossY{\Graph}{\BSet} \leq \epsR$ holds with some
desired probability. This is the main technical contribution of this
work.

Next, following in the footsteps of the construction of reliable
spanners, we use the one-dimensional construction to get
$(1+\eps)$-spanners that are $\epsR$-reliable either in expectation or
with good probability. The new constructions have size roughly
$\Of\bigl(n \log\!\log^2\!n \bigr)$.

\paragraph{Main idea.}
We borrow the notion of shadow from the work of Buchin \etal~\cite{bho-spda-20}. A point $\pa$
is in the $\alpha$-shadow if there is a neighborhood of $\pa$, such
that an $\alpha$-fraction of it belongs to the attack set.  One can think
about the maximum $\alpha$ such that $\pa$ is in the $\alpha$-shadow
of $\BSet$ as the \emph{depth} of $\pa$ (here, the depth is in the
range $[0,1]$).  A point with depth close to one, are intuitively
surrounded by failed points, and have little hope of remaining well
connected. Fortunately, only a few points have depth truly close to
one.  The flip side is that the attack has little impact on
shallow points (i.e., points with depth close to $0$). Similar to
people, shallow points are surrounded by shallow points. As such, only
a small fraction of the shallow points needs to be strongly connected
to other points in the graph, as paths from (shallow) points around them can
then travel via these hub points.

To this end, similar in spirit to skip lists, we define a random
gradation of the points
$\PS = \PS_0\supseteq \PS_1 \supseteq \ldots \supseteq \PS_{\log n}$,
where $|\PS_i| = n/2^i$ -- this is done via a random tournament
tree. In each level, each point of $\PS_i$ is connected to all its
neighbors within a certain distance (which increases as $i$
increases).  Intuitively, because of the improved connectivity, the
probability that a point is well-connected (after the attack)
increases if they belong to higher level of the gradation.  Thus, the
probability of a shallow point to remain well connected is,
intuitively, good. Specifically, we can quantify the probability of a
vertex to lose its connectivity as a function of its depth. Combining
this with bounds on the number of points of certain depths, results in
bounds on the expected size of the damaged set.

\paragraph{Reliable skip lists.}
Our construction can be interpreted as a reliable construction of skip lists.
Here, an attack removes certain cells in the skip list, which are no longer
available. This can happen, for example, if the skip list is stored
in a distributed fashion in a network, and certain nodes of the network
are down. Our construction implies that one can withstand an attack
with small expected loss. The previous work on skip graphs \cite{as-sg-07},
or \cite{bho-spda-20}, presented constructions of variants of
skip lists with somewhat similar properties, but using $\Of(n \log n)$ pointers.
The current construction requires only $\Of(n)$ pointers.

\paragraph{Comparison to previous work.} %
While we borrow some components of Buchin \etal \cite{bho-spda-20},
the basic scheme in the one-dimensional case, is new, and
significantly different -- the previous construction used expanders in
a hierarchical way.  The new construction requires different analysis
and ideas. The extension to higher dimension is relatively
straightforward and follows the ideas of Buchin \etal
\cite{bho-spda-20}, although some modifications and care are
necessary.

\paragraph{Why the oblivious model?} %
Previous work \cite{bho-spda-20} assumed an adaptive adversary. This
is somewhat unreasonable, as it is unlikely the adversary has full
information about the network. The ``price'' of such a strong
assumption is that the resulting constructions are complicated and
require a large number of edges. The alternative approach, used in
this work, is an \emphw{oblivious adversary} -- which is one of the
three standard adversarial models used in online algorithms
\cite{be-occa-98}.  In our settings, this corresponds to an adversary
that knows the point set, but not the edges used by the spanner. Here,
this obliviousness results in significantly simpler (and with fewer
edges) constructions of reliable spanners\footnote{Reality seems to be
   somewhat murkier -- the adversary has some partial information
   \cite{k-bpgsh-20}.}.

\paragraph{``Breaking'' the lower bound.}

As mentioned above, a reliable Euclidean spanner in the non-oblivious
model requires $\Omega( n \log n)$ edges \cite{bdms-rgs-13}. Thus, our
work here shows that one can break this lower bound by relaxing the
model. Beyond the improvement in the size of the spanner, we consider
this to be the most exciting property of our construction.

\paragraph{Paper organization.}
We review some necessary machinery in \secref{prelims}. The
one-\si{dimens}\-\si{ional} construction is described in
\secref{one_dim}. We describe the extension to higher dimensions in
\secref{higher:dim}.

\section{Preliminaries}
\seclab{prelims}

Let $\Graph=(\PS,\Edges)$ be a $t$-spanner for some $t\geq 1$.  An
\emphi{attack} on $\Graph$ is a set of vertices $\BSet$ that fail, and
no longer can be used. An attack is \emphi{oblivious}, if the set
$\BSet$ is picked without any knowledge of $\Edges$.

\begin{definition}[Reliable spanner]
    \deflab{formal:r:spanner}%
    Let $\Graph=(\PS,\Edges)$ be a $t$-spanner for some $t\geq 1$
    constructed by a (possibly) randomized algorithm.  Given a
    non-empty attack $\BSet$, its \emphi{damaged set} $\EBSet$ is a
    smallest set, such that for any pair of vertices
    $u,v \in \PS \setminus \EBSet$, we have
    \begin{equation*}
        \dGZ{\Graph \setminus \BSet}{u}{v} \leq t \cdot \dY{u}{v} ,
    \end{equation*}
    that is, $t$-paths are preserved for all pairs of points not
    contained in $\EBSet$.  The quantity
    $\cardin{\EBSet \setminus \BSet}$ is the \emphi{loss} of $\Graph$
    under the attack $\BSet$.  The \emphi{loss rate} of $\Graph$ is
    $\lossY{\Graph}{\BSet} = \cardin{\EBSet \setminus \BSet} /
    \cardin{\BSet}$. For $\epsR \in (0,1)$, the graph $\Graph$ is
    \emphi{$\epsR$-reliable} if $\lossY{\Graph}{\BSet} \leq \epsR$
    holds for any attack $\BSet \subseteq \PS$.

    Further, the random graph $\Graph$ is \emphi{$\epsR$-reliable in
       expectation} if $\Ex{\lossY{\Graph}{\BSet}} \leq \epsR$ holds
    for any oblivious attack $\BSet \subseteq \PS$.  For
    $\epsR,\epsP\in(0,1)$, we say that the graph $\Graph$ is
    \emphi{$\epsR$-reliable with probability $1-\epsP$} if
    $\Prob{\lossY{\Graph}{\BSet} \leq \epsR} \geq 1-\epsP$ holds for
    any oblivious attack $\BSet \subseteq \PS$.
\end{definition}

We emphasize that in the latter case the graph is random and the
expectation and the probability are taken with respect to the
distribution of the graphs.  In addition, the set $\EBSet$ is not
unique -- specifically, the attacker chooses the set $\BSet$, but we
have some choice in choosing the damaged set $\EBSet$ (such as,
absurdly, choosing $\EBSet$ to be the set of all the vertices of the
graph). In particular, the analysis we provide shoes that there is
always (in expectation or with good probability) a small damaged set.

\begin{definition}
    Let $\IRX{n}$ denote the \emph{interval} $\brc{ 1,\ldots, n}$.
    Similarly, for $x$ and $y$, let $\IRY{x}{y}$ denote the interval
    $\{x, x+1, \ldots, y\}$.
\end{definition}
We use the shadow notion as it was introduced by Buchin \etal
\cite{bho-spda-20}.

\begin{definition}
    Consider an arbitrary set $\BSet \subseteq \IRX{n}$ and a
    parameter $\alpha \in (0,1)$. A number~$i$ is in the \emphi{left
       $\alpha$-shadow} of $\BSet$, if and only if there exists an
    integer $j \geq i$, such that
    \begin{math}
        \cardin{\IRY{i}{j} \cap \BSet \bigr. }%
        \geq%
        \alpha \cardin{\IRY{i}{j} \bigr.}.
    \end{math}
    Similarly, $i$ is in the \emphi{right $\alpha$-shadow} of $\BSet$,
    if and only if there exists an integer $h$, such that $h \leq i$
    and
    \begin{math}
        \cardin{\IRY{h}{i} \cap \BSet} \geq \alpha
        \cardin{\IRY{h}{i}}.
    \end{math}
    The left and right $\alpha$-shadow of $\BSet$ is denoted by
    $\ShadowLY{\alpha}{\BSet}$ and $\ShadowRY{\alpha}{\BSet}$,
    respectively. The combined shadow is denoted by
    \begin{math}
        \ShadowY{\alpha}{\BSet}%
        =%
        \ShadowLY{\alpha}{\BSet} \cup \ShadowRY{\alpha}{\BSet}.
    \end{math}
\end{definition}

\begin{lemma}[\cite{bho-spda-20}]
    \lemlab{shadow}%
    For any set $\BSet \subseteq \IRX{n}$, and $\alpha \in (0,1)$, we
    have that
    $\cardin{\ShadowY{\alpha}{\BSet}} \leq \bigl(1 +
    2\ceil{1/\alpha}\bigr)\cardin{\BSet}$.
    Further, if $\alpha \in (2/3,1)$, we have that
    $\cardin{\ShadowY{\alpha}{\BSet}} \leq \cardin{\BSet} /
    (2\alpha-1)$.
\end{lemma}

\begin{definition}
    Given a graph $\Graph$ over $\IRX{n}$, a \emphi{monotone path}
    between $i,j \in \IRX{n}$, such that $i < j$, is a sequence of
    vertices $i = i_1 < i_2 < \cdots < i_k = j$, such that
    $i_{\ell-1}i_\ell \in \EdgesX{\Graph}$, for $\ell=2, \ldots, k$.
\end{definition}

A monotone path between $i$ and $j$ has length $|j-i|$.
Throughout the paper we use $\log x$ and $\ln x$ to denote the base $2$ and natural base logarithm
of $x$, respectively.  For any set
$A\subseteq \PS$, let $\compX{A} = \PS \setminus A$ denote the
complement of $A$.
For two integers $x,y > 0$, let~$\UpY{x}{y} = \ceil{ {x}/{y}}y $.

\section{Reliable spanners in one dimension}
\seclab{one_dim}%

We show how to build a random graph on $\IRX{n}$ that still has monotone
paths almost for all vertices that are not directly attacked. First,
in \secref{1-d-analysis}, we show that our construction is
$\epsR$-reliable in expectation. Then, in \secref{1-d-analysis-prob},
we show how to modify the construction to obtain a $1$-spanner that is
$\epsR$-reliable with probability $1-\epsP$.

\subsection{Construction}
\seclab{one_dim_constr}%

The input consists of a parameter $\epsR >0$ and the point set
$\PS=\IRX{n} = \brc{1,\ldots, n}$. The backbone of the construction is
a random elimination tournament, see \figref{tournament-tree} as an example.
We assume that $n$ is a power of $2$ as otherwise one can
construct the graph for the next power of two, and then throw away the
unneeded vertices.

The tournament is a full binary tree, with the leafs storing the
values from $1$ to $n$, say from left to right. The value of a node is
computed randomly and recursively. For a node, once the values of the
nodes were computed for both children, it randomly copies the value of
one of its children, with equal probability to choose either
child. Let $\PS_i$ be the values stored in the $i$\th bottom level of
the tree. As such, $\PS_0 =\PS$, and $\PS_{\log n}$ is a
singleton. Each set $\PS_i$ can be interpreted as an ordered set (from
left to right, or equivalently, by value).

Let
\begin{equation}
    \sp=1-\frac{\epsR}{8}
    \qquad\text{and}\qquad
    \eps=
    \frac{8(1-\sp)}{\constA \ln \epsR^{-1}}
    =%
    \frac{\epsR}{\constA \ln \epsR^{-1}},
    \eqlab{sp:value}%
\end{equation}
where $\constA>1$ is a sufficiently large constant.  Let $M$ be the
smallest integer for which $\cardin{\PS_M} \leq {2^{M/2}}/{\eps}$
holds (i.e., $M = \ceil{ (2/3)\log (\eps n)}$). For $i=0,1,\dots,M$,
and for all $\pa \in \PS_i$ connect $\pa$ with the first
\begin{equation}
    \connX{i} = \ceil{\frac{2^{i/2}}{\eps}}
    \eqlab{conn}%
\end{equation}
successors (and hence predecessors) of $\pa$ in $\PS_i$. Let $\Edges_i$ be
the set of all edges in level $i$. The graph $\Graph$ on $\PS$ is
defined as the union of all edges over all levels -- that is,
$\Edges(\Graph)=\cup_{i=0}^M \Edges_i$.  Note, that top level of the
graph $\Graph$ is a clique.

\begin{figure}
    \centering \includegraphics{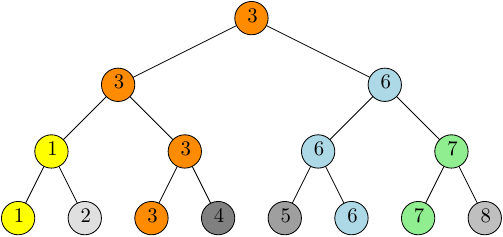}
    \caption{An example of a tournament tree with $n=8$.}
    \figlab{tournament-tree}
\end{figure}

\begin{remark}
    Before dwelling on the correctness of the construction, note that
    the obliviousness of the attack is critical. Indeed, it is quite
    easy to design an attack if the structure of $\Graph$ is known.
    To this end, let $\BSet_i$ be the set of
    $\connX{M} = O(n^{1/3} /\eps)$ values of $\PS_i$ closest to $n/2$
    -- namely, we are taking out the middle-part of the graph, that
    belongs to the $i$\th level.  Consider the attack
    $\BSet = \cup \BSet_i$. It is easy to verify that this attack
    breaks $\Graph$ into at least two disconnected graphs, each of
    size at least $n/2 - O(n^{1/3}\eps^{-1} \log n)$.
\end{remark}

\subsection{Analysis}
\seclab{1-d-analysis}
\begin{lemma}
    \lemlab{size-1d} The graph $\Graph$ has
    $\Of\left( n \epsR^{-1} \log \epsR^{-1} \right)$ edges.
\end{lemma}
\begin{proof}
    The number of edges contributed by a point in $\PS_i$ is at most
    $\connX{i}$ at level $i$, and $\cardin{\PS_i}= {n}/{2^i}$. Thus,
    we have
    \begin{align*}
        \cardin{\Edges(\Graph)}
        &\leq
        \sum_{i=0}^M |\PS_i| \cdot \connX{i}
        \leq
        \sum_{i=0}^M \frac{n}{2^i} \cdot \ceil{\frac{2^{i/2}}{\eps}}
        \leq
        \sum_{i=0}^M \frac{n}{2^i} \cdot \frac{2\cdot 2^{i/2}}{\eps}
        \leq
        \frac{n}{\eps} \cdot \sum_{i=0}^\infty \frac{2}{2^{i/2}}
        =
        \Of\left(\frac{n}{\eps}\right).
    \end{align*}%
\end{proof}

Fix an attack $\BSet \subseteq \PS$. The high-level idea is to show
that if a point $\pa \in \PS \setminus \BSet$ is far enough from the
faulty set, then, with high probability, there exist monotone paths
reaching far from $\pa$ in both directions. For two points
$\pa < \pb$, we show that if both $\pa$ and $\pb$ have far reaching
monotone paths, then the path going to the right from $\pa$, and the
path going to the left from $\pb$ must cross each other, which in turn
implies, that there is a monotone path between $\pa$ and
$\pb$. Therefore, it is enough to bound the number of points that does
not have far reaching monotone paths.

\begin{definition}[Stairway]
    Let $\pa\in \PS$ be an arbitrary point. The path
    $\pa=\pa_0,\pa_1,\dots,\pa_j$ is a right (resp., left)
    \emphi{stairway} of $\pa$ to level $j$, if
    \begin{enumerate}[label=(\roman*)]
        \item $\pa = \pa_0 \leq \pa_1 \leq \dots \leq \pa_j$ (resp.,
        $\pa \geq \pa_1 \geq \dots \geq \pa_j$),

        \item if $\pa_i \neq \pa_{i+1}$, then
        $ \pa_i \pa_{i+1} \in \Edges_i$, for $i=0,1,\dots,j-1$,

        \item $\pa_i \in \PS_i$, for $i=1,\dots,j$.
    \end{enumerate}
    Furthermore, a stairway is \emphi{safe} if none of its points are
    in the attack set $\BSet$. A right (resp., left) stairway is
    \emphi{usable}, if $\IRY{\pa_j}{n} \cap \PS_j$ (resp.,
    $\IRY{1}{\pa_j} \cap \PS_j$) forms a clique in $\Graph$.  Let
    $\StwSet \subseteq \PS$ denote the set of points that have a safe
    and usable stairway to both directions.
\end{definition}

Let $\sp_k = {\sp}/{2^k}$, for $k=0,1,\dots, \log n$. Let
$\SH_k = \ShadowY{\sp_k}{\BSet}$ be the $\sp_k$-shadow of $\BSet$, for
$k=0,1,\dots, \log n$. Observe that
$\SH_0 \subseteq \SH_1 \subseteq \cdots \subseteq \SH_{\log n}$, and
there is an index $j$ such that $\SH_j = \PS$, if
$\BSet \neq \emptyset$.  A point is classified according to when it
gets ``buried'' in the shadow. A point $\pa$, for $k\geq 1$, is a
\emphi{$k$\th round} point, if $\pa \in \SH_k \setminus
\SH_{k-1}$. Intuitively, a $k$\th round point is more likely to have a
safe stairway the larger the value of $k$ is.

\begin{definition}
    A point is \emphi{bad} if it belongs to $\BSet$, or it does not
    have a right or left stairway that is safe and usable. Formally, a
    point $\pa \in \PS$ is bad, if and only if
    $\pa \in \PS \setminus \StwSet$.
\end{definition}

\begin{lemma}
    \lemlab{mon-path} For any two points $\pa,\pb \in \StwSet$ that
    are not bad, there is a monotone path connecting $\pa$ and $\pb$
    in the residual graph $\Graph\setminus\BSet$.
\end{lemma}
\begin{proof}
    Suppose we have $\pa < \pb$. Let $(\pa,\pa_1,\dots, \pa_{j(\pa)})$
    be a safe usable right stairway starting from $\pa$ and
    $(\pb,\pb_1,\dots,\pb_{j(\pb)})$ be a safe usable left stairway
    from $\pb$. These stairways exist, since $\pa, \pb \in
    \StwSet$. Let $j=\min\pth{j(\pa),j(\pb)}$ and consider the
    stairways $(\pa,\pa_1,\dots,\pa_j)$ and
    $(\pb,\pb_1,\dots,\pb_j)$. Notice that both are safe and at least
    one of them is usable.

    Let $i$ be the first index such that $\pa_i \geq \pb_i$, if there
    is any. We distinguish two cases based on whether
    $\pa_i < \pb_{i-1}$ holds or not. In case $\pa_i < \pb_{i-1}$, the path
    $$(\pa,\pa_1,\dots,\pa_{i-1},\pa_i,\pb_{i-1},\dots,\pb_1,\pb)$$ is a
    monotone path from $\pa$ to $\pb$, since
    $\pb_i\pb_{i-1} \in \Edges_{i-1}$ implies
    $\pa_i\pb_{i-1} \in \Edges_{i-1}$. On the other hand, if we have
    $\pa_i \geq \pb_{i-1}$, the path
    $(\pa,\pa_1,\ldots,\pa_{i-1}, \pb_{i-1},\dots,\pb_1,\pb)$ is a
    monotone path between $\pa$ and $\pb$, since
    $\pa_{i-1} \pa_i \in \Edges_{i-1}$ implies
    $\pa_{i-1} \pb_{i-1} \in \Edges_{i-1}$.

    Finally, if $\pa_i < \pb_i$ holds for all
    $i=1,\dots,j$, then the path
    $(\pa,\pa_1,\dots,\pa_j,\pb_j,\dots,\pb_1,\pb)$ is a monotone path
    between $\pa$ and $\pb$. We have $\pa_j\pb_j \in \Edges_j$, since
    at least one of the stairways is usable. This concludes the proof
    that there is a monotone path from $\pa$ to $\pb$.
\end{proof}

\begin{lemma}
    \lemlab{hit:P:i}%
    For a fixed set $\QS \subseteq \IRX{n}$, we have that
    $\Prob{ \QS \cap \PS_i = \emptyset} \leq \exp( - \cardin{\QS}
    /2^i)$.
\end{lemma}
\begin{proof}
    Let $\QS = \{ \pb_1, \ldots, \pb_r\}$, and observe that knowing
    that certain points of $\QS$ are not in $\PS_i$, increases the
    probability of another point to be in $\PS_i$. That is,
    $\ProbCond{ \pb_j \in \PS_i}{\pb_1,\ldots, \pb_{j-1} \notin \PS_i}
    \geq \Prob{\pb_j \in \PS_i} = 1/2^i$.  As such, we have
    \begin{align*}
        \Prob{ \QS \cap \PS_i = \emptyset \bigr.}%
        &=%
        \Prob{\Bigl. \smash{ \bigcap_j (\pb_j \notin \PS_i)}} =%
        \prod_{j=1}^r \ProbCond{\pb_j \notin \PS_i }{\pb_1, \ldots,
           \pb_{j-1} \notin \PS_j} \\
        &\leq%
        \pth{1-1/2^i}^r%
        \leq %
        \exp( -r/2^i).~
    \end{align*}%
\end{proof}

\begin{figure}
    \centerline{\includegraphics{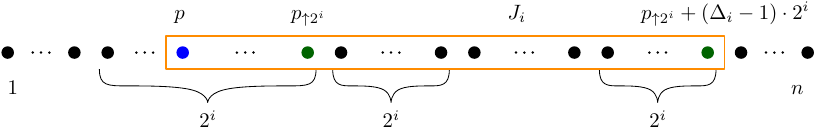}}%
    \caption{The interval
       $\IV_i=\IRY{\pa}{ \smash{\UpY{\pa}{2^i} + (\Delta_i -1) \cdot
             2^i} \bigr.}$.}
    \figlab{intervals}
\end{figure}

\begin{lemma}
    \lemlab{stairway}%
    Assume that $\epsR \in (0,1/2)$ and let
    $\pa \in \SH_k \setminus \SH_{k-1}$ be a $k$\th round point for
    some $k\geq 1$.  The probability that $\pa$ is bad is at most
    $(\epsR/2)^{k}/32$.
\end{lemma}
\begin{proof}
    For any integer $i \geq 1$, let
    $\Delta_i=\floor{{2^{\pth{i-1}/2}} / (2\eps) }$ and let
    \begin{equation*}
        \IV_i%
        =%
        \IRY{\pa}{ \smash{\UpY{\pa}{2^i} + (\Delta_i -1) \cdot 2^i}
           \bigr.},
    \end{equation*}
    see \figref{intervals}. Recall that $p \in \IRX{n}$, so
    $\UpY{\pa}{2^i} = \ceil{p/2^i}2^i$ is the next multiple of
    $2^i$. Let $\itop$ be the largest integer such that
    $\IV_\itop \subseteq \PS$.  For $i=0,\ldots, \itop$, the points of
    $\IV_{i+1} \cap \PS_{i}$ form a clique in $\Graph$, since
    \begin{equation*}
        \cardin{\IV_{i+1} \cap \PS_{i}}%
        \leq%
        \ceil{\cardin{\IV_{i+1}}/2^{i}}%
        \leq%
        \ceil{ 2^{i+1} \Delta_{i+1} /2^{i}}
        = %
        2 \Delta_{i+1}
        \leq%
        \ceil{\bigl.\smash{2^{i/2}}/{\eps}}
        =%
        \connX{i}.
    \end{equation*}
    Indeed, any two vertices of $\PS_i$ with distance at most
    $\connX{i}$ are connected by an edge of $\Edges_i$.  As such, it is
    enough to prove that there is a right safe stairway from $\pa$,
    that climbs on the levels to level $\itop$. Since
    $\IV_{\itop+1} \cap \PS_{\itop}$ forms a clique, it follows that
    such a stairway would be usable.

    Let $\Event_i$ be the event that
    $\pth{\IV_i \setminus \BSet} \cap \PS_i$ is empty, for
    $i=1,\dots,\itop$.  Since $\pa \notin \SH_{k-1}$, we have that
    \begin{math}
        \cardin{\IV_i \cap \BSet}%
        <%
        \sp_{k-1} \cardin{\IV_i}%
        \leq%
        2^{i} \sp_{k-1} \Delta_i.
    \end{math}
    On the other hand, we have
    \begin{math}
        \cardin{ \IV_i \cap \PS_i }%
        \geq%
        2^i (\Delta_i -1) /2^i%
        =%
        \Delta_i -1.
    \end{math}
    As such, if
    $\cardin{\IV_i \cap \BSet} < \cardin{ \IV_i \cap \PS_i }$ then
    $\Prob{\Event_i} = 0$. This happens if
    \begin{math}
        2^{i} \sp_{k-1} \Delta_i \leq \Delta_i-1%
        \iff%
        2^{i-k+1} \sp \leq (\Delta_i - 1)/ \Delta_i,
    \end{math}
    which happens if $i \leq k-2$, given that $\Delta_i \geq
    2$. Notice that $\Delta_i \geq 2$ holds for all $i\geq 1$, if
    $\eps \leq \frac{1}{4}$.

    So assume that $i \geq k-1$.  Let $\pb_1,\dots,\pb_r$ be all
    points of $\IV_i \setminus \BSet$, which are the possible
    candidates to be contained in
    $\pth{\IV_i \setminus \BSet} \cap \PS_i$.  By \Eqref{sp:value},
    there are at least
    \begin{align*}
      r
      = \cardin{\IV_i} - \cardin{\IV_i \cap \BSet} %
        &\geq%
        (1- \sp_{k-1}) \cardin{\IV_i} %
        \geq%
        (1- \sp_{k-1}) 2^{i} (\Delta_i -1)
        \geq%
        (1- \sp_{k-1}) 2^{i}
        \Bigl({\frac{2^{\pth{i-1}/2}}{2\eps}-2} \Bigr) \\
      &=%
        \frac{\constA(1- \sp_{k-1}) \ln \epsR^{-1}}{16(1-\sp)}
        2^{3i/2 -1/2} - (1- \sp_{k-1}) 2^{i+1} \\
        &\geq %
        \constA
        2^{3i/2 -9/2} \ln \epsR^{-1} - 2^{i+1}
    \end{align*}
    such points. Observe, that by the structure of the construction, a
    point is more likely to be contained in $\PS_i$ conditioned on the
    event there are some other points which are not contained in
    $\PS_i$. Therefore, by \lemref{hit:P:i}, we have
    \begin{math}
        \Prob{\bigl.\Event_i} \leq \exp \bigl( - {r}/{2^i} \bigr) \leq
        \tau_i,
    \end{math}
    for
    \begin{math}
        \tau_i = \exp\pth{2 - \constA 2^{i/2 -9/2} \ln \epsR^{-1}}.
    \end{math}
    The sequence $\tau_i$ has a fast decay in $i$, since
    \begin{align*}
      \frac{\tau_{i+1}}{\tau_i}
      &=
        \exp \pth{-(\sqrt{2}-1) \constA 2^{i/2 - 9/2} \ln \epsR^{-1}}
        \leq
        \exp \pth{-\constA 2^{-6} \ln 2}
        =
        2^{-\constA 2^{-6}}
        \leq
        \frac{1}{2},
    \end{align*}
    if $\constA\geq 2^6$ holds.  Thus, we have
    \begin{align*}
      \Prob{\cup_{i=1}^\itop \Event_i\bigr.}%
      &\leq%
        \sum_{i=1}^\itop \Prob{ \Event_i\bigr.}%
        \leq%
        \sum_{i=k-1}^\itop \tau_i%
        \leq%
        2 \tau_{k-1}
        =%
        2\exp\pth{ 2 - \constA 2^{(k-1)/2 -9/2} \ln \epsR^{-1}} \\
      &\leq%
        16\exp\pth{  - \frac{\constA}{32} 2^{k/2} \ln \epsR^{-1}}
        =
        16 \cdot \epsR^{({\constA}/{32}) 2^{k/2}}
        \leq
        2^4 \cdot \epsR^{({\constA}/{2^6}) k} \\
      &\leq
        2^4 \cdot \pth{\frac{1}{2}}^{({\constA}/{2^7}) k} \cdot
        \pth{\epsR^{({\constA}/{2^7})}}^k
        \leq%
        \frac{(\epsR/2)^{k}}{64}
    \end{align*}
    for $\constA \geq 2^{11}$, using the conditions
    $0<\epsR\leq \frac{1}{2}$, $k\geq 1$ and the fact that
    $x\leq 2^x$.

    Let $\pa_i$ be the leftmost point in
    $(\IV_i \setminus \BSet) \cap \PS_i$, for $i\geq 0$. Since
    $\PS_i \subseteq \PS_{i-1}$, for all $i$, it follows that
    $\pa = \pa_0 \leq \pa_1 \leq \cdots \leq \pa_\itop$. Furthermore,
    since $\IV_{i+1} \cap \PS_{i}$ is a clique in level $i$ of $\Graph$, and
    $\pa_i, \pa_{i+1} \in \IV_{i+1} \cap \PS_{i}$, it follows that
    $\pa_i \pa_{i+1} \in \Edges_i$, if $\pa_i \neq \pa_{i+1}$,
    for all $i$. We conclude that $\pa, \pa_1, \ldots, \pa_\itop$ is a
    safe and usable right stairway in $\Graph$.

    The bound now follows by applying the same argument symmetrically
    for the left stairway. Indeed, using the union bound, we obtain
    $\Prob{ \pa \text{ is bad}} \leq 2 {(\epsR/2)^{k}}/{64} =
    {(\epsR/2)^{k}}/{32}$.
\end{proof}

\begin{lemma}
    \lemlab{bad-points}%
    Let $\epsR\in (0,1/2)$ and $\BSet\subseteq \PS$ be an oblivious
    attack. Recall, that $\compX{\StwSet}$ is the set of bad
    points. Then, we have
    $\Ex{\cardin{\compX{\StwSet}}} \leq (1+\epsR) \cardin{\BSet}$.
\end{lemma}
\begin{proof}
    We may assume that all the points of $\SH_0$ are bad. Fortunately,
    by \lemref{shadow}, we have
    $\cardin{\SH_0} \leq \cardin{\BSet} / (2\sp -1) = \cardin{\BSet} /
    (1- \epsR/4) \leq (1+\epsR/2) \cardin{\BSet}$, since
    $\sp=1-{\epsR}/{8}$ and $1/(1-x/4) \leq 1 + x/2$ for $0 \leq x \leq 2$.
    Again, using \lemref{shadow}, we have
    \begin{equation*}
        \cardin{\SH_k \setminus \SH_{k-1}}%
        \leq%
        \cardin{\SH_k} \leq \pth{1 + 2\ceil{2^k/\sp}}\cardin{\BSet} \leq
        \pth{3 + \frac{2^{k+1}}{\sp}}\cardin{\BSet} %
        \leq%
        2^{k+3} \cardin{\BSet}.
    \end{equation*}
    For $k\geq 1$, we have, by \lemref{stairway}, that
    \begin{align*}
      b_k
      &=
        \Ex{\bigl.\cardin{ \pth{\SH_k \setminus \SH_{k-1}} \cap \compX{\StwSet}}}
        \leq %
        \sum_{\pa \in \SH_k \setminus \SH_{k-1}} \Prob{ \pa \text{ is bad}} %
        \leq%
        2^{k+3} \cardin{\BSet} \cdot
        \frac{(\epsR/2)^{k}}{32}
        \leq
        \frac{\epsR^k}{4} \cardin{\BSet}.
    \end{align*}
    Since,
    $\compX{\StwSet} = \pth{\SH_0 \cap \compX{\StwSet}} \cup
    \bigcup_{k\geq 1} \bigl[ \pth{\SH_k \setminus \SH_{k-1}} \cap
    \compX{\StwSet} \bigr] $, we have, by linearity of expectation,
    that
    \begin{align*}
      \frac{\Ex{\cardin{\compX{\StwSet}}\bigr.}}{\cardin{\BSet}}%
      &\leq%
        \frac{1}{\cardin{\BSet}} \pth{\cardin{\SH_0} + \sum_{k=1}^\infty b_k}
        \leq%
        1+\frac{\epsR}{2} +
        \sum_{k=1}^\infty \frac{\epsR^k}{4}
        \leq%
        1+\frac{\epsR}{2} + \frac{\epsR}{4(1-\epsR)}
        \leq%
        \pth{1+\epsR},
    \end{align*}
    since $\epsR < 1/2$.
\end{proof}

\begin{theorem}
    \thmlab{1-d}%
    Let $\epsR \in \pth{0,1/2}$ and $\PS=\IRX{n}$ be fixed.  The graph
    $\Graph$, constructed in \secref{one_dim_constr}, has
    $ \Of\pth{n \epsR^{-1} \log\epsR^{-1}}$ edges, and it is a
    $\epsR$-reliable $1$-spanner of $\PS$ in expectation. Formally,
    for any oblivious attack $\BSet$, we have
    $\Ex{\lossY{\Graph}{\BSet}} \leq \epsR$.
\end{theorem}
\begin{proof}
    By \lemref{size-1d} the size of the construction is
    \begin{math}
        \cardin{\Edges(\Graph)} %
        =%
        \Of\pth{n \epsR^{-1} \log\epsR^{-1}}.
    \end{math}
    Let $\BSet\subseteq \PS$ be an oblivious attack and consider the
    bad set $\PS\setminus \StwSet$.  By \lemref{mon-path}, for any two
    points outside the bad set, there is a monotone path connecting
    them. Further, by \lemref{bad-points}, we have
    $\Ex{\cardin{\PS\setminus\StwSet}} \leq (1+\epsR) \cardin{\BSet}$
    for any oblivious attack. Therefore, we obtain
    $\Ex{\lossY{\Graph}{\BSet}} \leq \Ex{\cardin{\compX{\StwSet}
          \setminus \BSet} / \cardin{\BSet}} \leq \epsR$.
\end{proof}

\subsection{Probabilistic bound}
\seclab{1-d-analysis-prob}

One can replace the guarantee, in \thmref{1-d}, on the bound of the
loss rate (which holds in expectation), by an upper bound that holds
with probability at least $1-\epsP$, for some prespecified
$\epsP > 0$. A straightforward application of Markov's inequality
implies that taking the union of $\log \epsP^{-1}$
independent copies ($\Graph'$) of the construction of \thmref{1-d} with parameter $\epsR/2$, results in a
graph with the desired property. Indeed, we have
\begin{align*}
    \Prob{\lossY{\Graph}{\BSet} > \epsR}
    &\leq
    \Prob{\lossY{\Graph'}{\BSet} > \epsR}^{\log\epsP^{-1}}
    \leq
    \pth{\frac{\Ex{\lossY{\Graph'}{\BSet}}}{\epsR}}^{\log\epsP^{-1}}
    \leq
    \pth{\frac{1}{2}}^{\log\epsP^{-1}}
    =
    \epsP.
\end{align*}
Here we show how one can do better to avoid the multiplicative factor $\log\epsP^{-1}$.

\paragraph{Construction.} %
The input consists of two parameters $\epsR,\epsP > 0$ and the set
$\PS = \IRX{n}$. Let $\Graph$ be the graph constructed in
\secref{one_dim_constr} with parameters
\begin{equation*}
    \sp=1-\frac{\epsR}{8}
    \qquad\text{and}\qquad
    \eps=
    \frac{8(1-\sp)}{\constA (\ln \epsR^{-1} + \ln \epsP^{-1})}
    =%
    \frac{\epsR}{\constA (\ln \epsR^{-1} + \ln \epsP^{-1})},
\end{equation*}
where $\constA>1$ is a sufficiently large constant. First, we need a
variant of \lemref{stairway} to bound the probability of a $k$\th
round point being bad, using the new value of $\eps$.
\begin{lemma}
    \lemlab{stairway_2}%
    Assume that $\epsR \in (0,1/2)$, $\epsP \in (0,1)$ and let
    $\pa \in \SH_k \setminus \SH_{k-1}$ be a $k$\th round point for
    some $k\geq 1$.  The probability that $\pa$ is bad is at most
    ${\epsR \cdot \epsP}/{2^{3k+4}}$.
\end{lemma}
\begin{proof}
    The proof is the same as the proof of \lemref{stairway}. The only
    difference is due to the new value of $\eps$, which results in
    $$\tau_i = \exp\pth{2 - \constA 2^{i/2 -9/2} (\ln \epsR^{-1} + \ln
       \epsP^{-1})},$$ using the same notation. Therefore, we have
    \begin{align*}
      \Prob{p \in \SH_k \setminus \SH_{k-1} \text{ is bad}}
      &\leq
        4 \tau_{k-1}
        =
        4\exp\pth{2 - \constA 2^{k/2 -5} (\ln \epsR^{-1} + \ln \epsP^{-1})} \\
      &\leq
        2^5\exp\pth{ - \frac{\constA}{2^6} k (\ln \epsR^{-1} + \ln \epsP^{-1})}
        =
        2^5 \cdot \epsR^{\frac{\constA}{2^6}k} \cdot  \epsP^{\frac{\constA}{2^6}k} \\
      &\leq
        2^5 \cdot \pth{\frac{1}{2}}^{({\constA}/{2^6})k-1} \cdot
        \epsR \cdot \epsP
        =
        2^{-\frac{\constA}{2^6}k +6} \cdot \epsR \cdot \epsP
        \leq
        2^{-3k -4} \cdot \epsR \cdot \epsP,
    \end{align*}
    for $\constA\geq 2^{10}$. See \lemref{stairway} for a complete
    proof.
\end{proof}

\begin{lemma}
    \lemlab{bad-points-prob}%
    Let $\epsR\in (0,1/2)$, $\epsP\in (0,1)$ be fixed and
    $\BSet\subseteq \PS$ be an oblivious attack.  Then, with
    probability $\geq 1-\epsP$, the number of bad points is at most
    $(1+\epsR)\cardin{\BSet}$. That is, we have
    \begin{math}
        \Prob{\cardin{\compX{\StwSet}} \leq (1+\epsR) \cardin{\BSet}}
        \geq 1-\epsP.
    \end{math}
\end{lemma}

\begin{proof}
    The idea is to give bounds on the number of bad
    $k$\th round
    points for all $k\geq 1$.  Let $\Event_k$ be the event that
    $\cardin{ \pth{\SH_k \setminus \SH_{k-1}} \cap \compX{\StwSet}} >
    \frac{\epsR}{2^{k+1}} \cardin{\BSet}$ happens, for $k\geq
    1$. Recall, by the choice of $\sp$, we have
    $\cardin{\SH_0 \cap \compX{\StwSet}} \leq \cardin{\SH_0} \leq
    \pth{1+\frac{\epsR}{2}}\cardin{\BSet}$. Notice, that at least one
    of the events $\Event_k$ must happen, for $k\geq 1$, in order to
    have $\cardin{\compX{\StwSet}} > (1+\epsR)\cardin{\BSet}$, since
    \begin{align*}
        \cardin{\compX{\StwSet}}
        &=
        \cardin{\SH_0 \cap \compX{\StwSet}} + \sum_{k=1}^\infty \cardin{ \pth{\SH_k \setminus \SH_{k-1}} \cap \compX{\StwSet}}
        \leq
        \pth{1+\frac{\epsR}{2}}\cardin{\BSet} + \sum_{k=1}^\infty \frac{\epsR}{2^{k+1}} \cardin{\BSet}
        =
        (1+\epsR)\cardin{\BSet}.
    \end{align*}
    Using Markov's inequality and \lemref{stairway_2} we get
    \begin{align*}
      \Prob{\Event_k}
      \leq
        \frac{\Ex{\cardin{ \pth{\SH_k \setminus \SH_{k-1}} \cap \compX{\StwSet}}}}{\frac{\epsR}{2^{k+1}} \cardin{\BSet}}
        &\leq
        \frac{\cardin{\SH_k}\cdot \Prob{p \in \SH_k \setminus \SH_{k-1} \text{ is bad}}}{\frac{\epsR}{2^{k+1}} \cardin{\BSet}}
        \leq
        \frac{2^{k+3} \cardin{\BSet}\cdot \frac{\epsR \cdot \epsP}{2^{3k+4}}}{\frac{\epsR}{2^{k+1}} \cardin{\BSet}}
        =
        \frac{\epsP}{2^k}.
    \end{align*}
    Therefore, we obtain
    \begin{align*}
      \Prob{\cardin{\compX{\StwSet}} > (1+\epsR)\cardin{\BSet} \bigr.}
      &\leq
        \Prob{\cup_{k\geq 1} \Event_k \bigr.}
        \leq
        \sum_{k=1}^\infty \Prob{\Event_k \bigr.}
        \leq
        \sum_{k=1}^\infty \frac{\epsP}{2^k}
        \leq
        \epsP,
    \end{align*}
    which is equivalent to
    $\Prob{\cardin{\compX{\StwSet}} \leq (1+\epsR) \cardin{\BSet}}
    \geq 1-\epsP$.
\end{proof}

\begin{theorem}
    \thmlab{1-d-prob}%
    Let $\epsR \in \pth{0,1/2}$, $\epsP \in \pth{0,1}$ and
    $\PS=\IRX{n}$ be fixed.  The graph $\Graph$, constructed above, is
    a $\epsR$-reliable $1$-spanner of $\PS$, with probability at least
    $1 -\epsP$. Formally, we have
    $\Prob{\lossY{\Graph}{\BSet} \leq \epsR} \geq 1-\epsP$ for any
    oblivious attack $\BSet$. Furthermore, the graph $\Graph$ has
    $ \Of\pth{n \epsR^{-1} (\log\epsR^{-1} + \log\epsP^{-1})}$ edges.
\end{theorem}
\begin{proof}
    The bound on the size follows directly from \lemref{size-1d}.
    Let $\BSet\subseteq \PS$ be an oblivious attack and consider the
    bad set $\PS\setminus \StwSet$.  By \lemref{mon-path}, for any two
    points outside the bad set, there is a monotone path connecting
    them. Further, by \lemref{bad-points-prob}, we have
    $\Prob{\lossY{\Graph}{\BSet} \leq \epsR} \geq \Prob{\cardin{\compX{\StwSet}} \leq (1+\epsR) \cardin{\BSet}}
    \geq 1-\epsP$ for any oblivious attack.
\end{proof}

\section{Reliable spanners in higher dimensions}
\seclab{higher:dim}

Now we turn to the higher-dimensional setting, and show that one can
construct spanners with near linear size that are reliable in
expectation or with some fixed probability (which can be provided as
part of the input).  We use the same technique as Buchin \etal
\cite{bho-spda-20}, that is, we use our one-dimensional construction
as a black box in combination with a result of
Chan~\etal~\cite{chj-lsota-20}. Let the dimension $d>1$ be fixed. In
the following we assume $\PS \subset [0,1)^d$, which can be achieved
by an appropriate scaling and translation of the $d$-dimensional
Euclidean space $\Re^d$. For an ordering $\order$ of $[0,1)^d$, and
two points $\pp,\pq \in [0,1)^d$, such that $\pp \prec \pq$, let
$(\pp,\pq)_{\order} = \Set{\pz \in [0,1)^d}{ \pp \prec \pz \prec \pq}$
be the set of points between $\pp$ and $\pq$ in the order $\order$.

\begin{theorem}[\cite{chj-lsota-20}]
    \thmlab{lso}%
    For $\epsB \in (0,1)$, there is a set $\ordAll(\epsB)$ of
    $M(\epsB) = \Of( \epsB^{-d} \log \epsB^{-1} )$ orderings of
    $[0,1)^d$, such that for any two (distinct) points
    $\pp, \pq \in [0,1)^d$, with $\ell = \dY{\pp}{\pq}$, there is an
    ordering $\order \in \ordAll$, and a point $\pz \in [0,1)^d$, such
    that %
    \begin{enumerate}[label=(\roman*)]
            \item $\pp \prec_\order \pq$,
            \item
            $(\pp,\pz)_\order \subseteq \ballY{\pp}{\bigl.  \epsB
               \ell}$,
            \item
            $(\pz,\pq)_\order \subseteq \ballY{\pq}{\bigl.  \epsB
               \ell}$, and%
            \item $\pz \in \ballY{\pp}{\bigl.  \epsB \ell}$ or
            $\pz \in \ballY{\pq}{\bigl.  \epsB \ell}$.
    \end{enumerate}
    \noindent%
    Furthermore, given such an ordering $\order$, and two points
    $\pp,\pq$, one can compute their ordering, according to $\order$,
    using $\Of(d\log \epsB^{-1})$ arithmetic and bitwise-logical
    operations.
\end{theorem}

The above theorem ensures that it is enough to maintain only a ``few''
linear orderings, and for any pair of points $p,q \in \PS$ there
exists an ordering where all points that lie between $p$ and $q$ are
either very close to $p$ or $q$. It is natural to build the
one-dimensional construction for each of these orderings with some
carefully chosen parameter. Then, since there is a reliable path in
the one-dimensional construction, there is an edge $p'q'$ along the
path between $p$ and $q$ that connects the locality of $p$ and the
locality of $q$. We fix the edge $p'q'$ and apply recursion on the
subpaths from $p$ to $p'$ and $q$ to $q'$.

\subsection{Construction}
\seclab{higher_dim_constr}%
Let $\epsR,\eps \in \pth{0,1}$ be fixed parameters and
$\PS \subseteq [0,1)^d$ be a set of $n$ points. Set
$\epsB = \eps/16$ in \thmref{lso} and let
$\ordAll=\ordAll(\epsB)$ be the set of $M=M\pth{\epsB}$ orderings that
fulfills the conditions of the theorem. We define
$\epsR' = \frac{\epsR}{3MN}$, where $N=\ceil{\log\!\log n}$. Now, for
each ordering $\order \in \ordAll$, we build $N$ independent spanners
$\Graph^1_\order,\dots,\Graph^N_\order$, using the construction in
\secref{one_dim_constr} with parameter $\epsR'$. The (random) graph $\Graph$ is
defined as the union of graphs $\Graph^i_\order$ for all
$\order \in \ordAll$ and $i\in \IRX{N}$, that is,
$\Edges\pth{\Graph} = \cup_{\order\in\ordAll, i \in \IRX{N}} \Edges\pth{\Graph^i_\order}$.

\subsection{Analysis}
\seclab{higher_dim_anal}%

\begin{lemma}
    \lemlab{size-d-dim}%
    The graph $\Graph$, constructed above, has
    $\Of\pth{\constA \, n \log\!\log^2\!n \log\!\log\!\log n}$ edges, where the
    $\Of$ hides constant that depends on the dimension $d$, and
    \begin{math}
        \constA%
        =%
        \Of(\eps^{-2d}\epsR^{-1} \log^3 \eps^{-1}\log \epsR^{-1}).
    \end{math}
\end{lemma}
\begin{proof}
    There are $M = \Of\pth{\eps^{-d} \log\eps^{-1}}$ orderings,
    and for each ordering there are $N$ copies, for which we build the
    one-dimensional construction with parameter $\epsR'$. The size of
    the one-dimensional construction is
    $\Of\pth{n \cdot \epsR'^{-1} \cdot \log\epsR'^{-1}}$, by
    \lemref{size-1d}. Therefore, $\Graph$ has size
    \begin{align*}
      \cardin{\Edges\pth{\Graph}}
      &=
        \cardin{\cup_{\order\in\ordAll,i\in\IRX{N}} \Edges\pth{\Graph^i_\order}}
        \leq
        \sum_{\order\in\ordAll,i\in\IRX{N}} \cardin{\Edges\pth{\Graph'_\order}}
      \leq
      N M \cdot \Of\pth{n \cdot \epsR'^{-1} \cdot \log\epsR'^{-1}}%
      \\
      &
        =
        \Of\pth{n \cdot N^2 M^2 \epsR^{-1} \cdot \pth{\log\epsR^{-1}
        + \log N + \log M }} \\
      &=
        \Of\bigl(n \log\!\log^2\!n \cdot
        \eps^{-2d}\log^2\eps^{-1}
        \cdot \epsR^{-1} (\log\epsR^{-1} +
        \log\!\log\!\log n +
        d\log
        \eps^{-1} +\log\!\log \eps^{-1} )\bigr) \\
      &=
        \Of\pth{ \constA \, n \log\!\log^2\!n \log\!\log\!\log n},
    \end{align*}
    where $\constA=\Of(\eps^{-2d}\epsR^{-1} \log^3 \eps^{-1}\log \epsR^{-1})$.
\end{proof}

Fix an attack set $\BSet \subseteq \PS$. In order to bound
$\lossY{\Graph}{\BSet}$ in expectation, we define a sequence of sets
$\BSet_0 \subseteq \BSet_1 \subseteq \dots \subseteq \BSet_N$ as follows. First, we set $\BSet_0 = \BSet$. Then, for
$i=1,\dots,N$, we define $\BSet_i^\order$ for each $\order\in\ordAll$
to contain all points that do not have a right or left stairway in
$\Graph^i_\order$ that is safe and usable with respect to
$\BSet_{i-1}$, that is, $\BSet_i^\order$ contains the bad points with
respect to $\BSet_{i-1}$. We set
$\BSet_i = \cup_{\order\in\ordAll} \BSet_i^\order$. Our goal is to
show that the expected size of $\BSet_N$ is small, and there is a
$(1+\eps)$-path for all pairs of points outside of $\BSet_N$.

\begin{lemma}
    \lemlab{bad-iter} Let $\BSet$ be an oblivious attack and let
    $\BSet_0 \subseteq \BSet_1 \subseteq \dots \subseteq \BSet_N$ be
    the sequence defined above. Then, for $i=1,\dots,N$, we have
    $\Ex{\cardin{\BSet_i^\order} \; | \; \BSet_{i-1}} \leq
    \pth{1+\epsR'}\cardin{\BSet_{i-1}}$, for all $\order \in \ordAll$.
\end{lemma}
\begin{proof}
    The set $\BSet_{i-1}$ has information only about
    graphs $\Graph^j_\order$ for $j\leq i-1$. Thus, the attack
    $\BSet_{i-1}$ on the graph $\Graph^i_\order$ is oblivious and we
    have
    $\Ex{\cardin{\BSet_i^\order} \; | \; \BSet_{i-1}} \leq
    \pth{1+\epsR'}\cardin{\BSet_{i-1}}$ by~\lemref{bad-points}.%
\end{proof}

\begin{lemma}
    \lemlab{bad-points-d-dim} Let $\BSet_N$ be the set defined
    above. For any oblivious attack $\BSet$, the expected size of
    $\BSet_N$ is at most $\pth{1+\epsR}\cdot \cardin{\BSet}$.
\end{lemma}

\begin{proof}
    By \lemref{bad-iter} we have
    $\Ex{\cardin{\BSet_i^\order} \; | \; \BSet_{i-1}} \leq
    \pth{1+\epsR'}\cardin{\BSet_{i-1}}$ for all $\order \in
    \ordAll$. Therefore,
    \begin{align*}
      \ExCond{\bigl.\cardin{\BSet_i}}{ \BSet_{i-1}}
      &\leq
        \bigl(\pth{1+\epsR'}\cardin{\BSet_{i-1}} - \cardin{\BSet_{i-1}}\bigr) \cdot M + \cardin{\BSet_{i-1}}
        =
        \pth{1+\frac{\epsR}{3N}}\cardin{\BSet_{i-1}}
    \end{align*}
    holds, for $i=1,\dots,N$, which gives
    \begin{align*}
      \Ex{\cardin{\BSet_N} \bigr.}
      \leq
      \Ex{\Ex{\cardin{\BSet_N} \; | \; \BSet_{N-1}}}
      &\leq
        \pth{1+\frac{\epsR}{3N}} \cdot \Ex{\cardin{\BSet_{N-1}}} \\
        &\leq
        \pth{1+\frac{\epsR}{3N}}^N \cdot \Ex{\cardin{\BSet_0}}
        =
        \pth{1+\frac{\epsR}{3N}}^N \cdot \cardin{\BSet}.
    \end{align*}
    Using $1+x \leq e^x \leq 1+3x$, for $x\in[0,1]$, we obtain
    \begin{align*}
      \Ex{\cardin{\BSet_N}}
      &\leq
        \pth{1+\frac{\epsR}{3N}}^N \cdot \cardin{\BSet}
        \leq
        \exp{\pth{N \frac{\epsR}{3N}}} \cdot \cardin{\BSet}
        =
        e^\frac{\epsR}{3} \cdot \cardin{\BSet}
        \leq
        \pth{1+\epsR} \cdot \cardin{\BSet}.
        \qedhere
    \end{align*}
\end{proof}

\begin{lemma}
    \lemlab{1+eps-path}%
    Let $\BSet_N$ be the set defined above. Then, for any two points
    $p,q\in \PS\setminus\BSet_N$, there is a $(1+\eps)$-path in the
    graph $\Graph \setminus \BSet$.
\end{lemma}
\begin{proof}
    The proof is essentially the same as the proof
    of Theorem~4.3 in \cite{bho-spda-20}.

    Let $p,q \in \PS \setminus \BSet_N$ be fixed.  According to
    \thmref{lso}, there is an ordering $\order\in\ordAll$, such that
    all the points $z\in \pth{p,q}_\order$ lie in one of the balls of
    radius $\epsB \dY{p}{q}$ around $p$ and $q$. Recall that the graph
    $\Graph$ contains $\Graph^N_\order$ as a subgraph. Since
    $p,q \notin \BSet_N$ and $\Graph^N_\order$ is reliable, there is a
    path connecting $p$ and $q$ that is monotone with respect to
    $\order$ and avoids any point in $\BSet_{N-1}$ by
    \thmref{1-d}. Therefore, there is a unique edge $p'q'$ along
    this path such that $p'$ is in the close neighborhood of $p$ and
    $q'$ is in the close neighborhood of $q$. Furthermore, we also
    have that $p',q' \in \PS \setminus \BSet_{N-1}$. We fix the edge
    $p'q'$ in path $\pi$ and find subpaths between the pairs
    $pp'$ and $qq'$ in a recursive manner. The bounds on the
    distances are %
    \begin{enumerate}[label=(\roman*)]
        \item
        \begin{math}
            \dY{p'}{q'}%
            \leq%
            \pth{1 + 2\epsB}\dY{p}{q},
        \end{math}

        \item $\dY{p}{p'} \leq \epsB\dY{p}{q}$ and similarly
        $\dY{q}{q'} \leq \epsB\dY{p}{q}$.
    \end{enumerate}
    We repeat this process $N-1$ times. Let $Q_i$ be the set of pairs
    that needs to be connected in the $i$\th round, that is,
    $Q_0=\{pq\}, Q_1=\{pp',qq'\}$ and so on. There are
    at most $2^i$ pairs in $Q_i$ and for any pair $xy\in Q_i$ we
    have $x,y \in \PS\setminus \BSet_{N-i}$. For each pair
    $xy\in Q_i$, there is an ordering $\order$ such that the
    argument above can be repeated. That is, there is a monotone path
    in the graph $\Graph^{N-i}_\order \setminus \BSet_{N-i-1}$
    according to $\order$ and there is an edge $x'y'$ along this
    path such that %
    \begin{enumerate}[label=(\roman*)]
        \item
        \begin{math}
            \dY{x'}{y'}%
            \leq%
            \pth{1 + 2\epsB}\dY{x}{y}%
            \leq%
            \pth{1 + 2\epsB}\epsB^{i}\dY{p}{q},
        \end{math}

        \item
        $\dY{x}{x'} \leq \epsB\dY{x}{y} \leq
        \epsB^{i+1}\dY{p}{q}$ and similarly
        $\dY{y}{y'} \leq \epsB^{i+1}\dY{p}{q}$.
    \end{enumerate}
    The edge $x'y'$ is added to path $\pi$ and the pairs $xx'$ and
    $yy'$ are added to $Q_{i+1}$, unless they are trivial (i.e.,
    $x=x'$ or $y=y'$). After $N-1$ rounds, $Q_{N-1}$ is the set of
    active pairs that still needs to be connected. Notice that
    $x,y\in \PS\setminus \BSet_1$ holds for any pair $xy\in
    Q_{N-1}$. Again, for each pair in $Q_{N-1}$, we apply \thmref{lso}
    and \thmref{1-d} to obtain a monotone path according to some
    ordering $\order$ in the graph $\Graph^1_\order$. None of these
    paths use any points in $\BSet$. In order to complete the path
    $\pi$ we add the whole paths obtained in the last step. It is not
    hard to see that the number of edges of each of the paths added in
    the last step is at most $2\log n$. Indeed, it is clear from the
    analysis of our one-dimensional construction that a path using the
    stairways can have at most two points per level. Since the number
    of levels in the construction is fewer than $\log n$, we get the
    bound $2\log n$.

    Now, that we have a path $\pi$ that connects the points $p$ and
    $q$ without using any points in the failed set $\BSet$, we give an
    upper bound on the length of $\pi$. First, we calculate the total
    length added in the last step. There are
    $\cardin{Q_{N-1}} \leq 2^{N-1}$ pairs in the last step and for
    each pair $xy \in Q_{N-1}$ we have
    $\dY{x}{y} \leq \dY{p}{q}\epsB^{N-1}$. Thus, we obtain
    \begin{align*}
      \sum_{\{x,y\} \in Q_{N-1}} & \text{length} \pth{\pi[x,y]}
      \leq
        2^{N-1} \pth{\pth{1+2\epsB} \dY{p}{q} \epsB^{N-1} + 2\log n \dY{p}{q}\epsB^N } \\
      &\leq
        2 \cdot 2\epsB \dY{p}{q} + \pth{2\epsB}^N \log n \dY{p}{q}
        =
        \pth{4\epsB + \pth{2\epsB}^N \log n} \dY{p}{q} \\
      &\leq
        \pth{\frac{\eps}{4} + \pth{\frac{\eps}{8}}^{\log\!\log n} \log n} \dY{p}{q} \\
      &\leq
        \pth{\frac{\eps}{4} + \frac{\eps}{4} \cdot \pth{\frac{1}{2}}^{\log\!\log n} \log n} \dY{p}{q}
        =
        \frac{\eps}{2} \dY{p}{q},
    \end{align*}
    where we simply use $2\epsB \leq 1$ in the second line and
    $\epsB=\eps/16$ and $N=\ceil{\log\!\log n}$ in the third
    line.  Second, we bound the total length of the edges that were
    added to path $\pi$ in any round except the last. This contributes
    at most
    \begin{align*}
      \sum_{i=0}^{N-2} 2^i \cdot \pth{1+2\epsB} \epsB^i \dY{p}{q}
      &\leq
        \pth{1+2\epsB} \dY{p}{q} \cdot \sum_{i=0}^{\infty}\pth{2\epsB}^i
        =
        \pth{1+2\epsB} \dY{p}{q} \cdot \frac{1}{1-2\epsB} \\
      &=
        \pth{1+ \frac{4\epsB}{1-2\epsB}} \dY{p}{q}
        =
        \pth{1+ \frac{\eps/4}{1-\eps/8}} \dY{p}{q} \\
        &\leq
        \pth{1+\frac{\eps}{2}} \dY{p}{q}
    \end{align*}
    to the length of $\pi$. Therefore the total length of the path
    $\pi$ connecting $p$ and $q$, without using any points of $\BSet$,
    is at most $\pth{1+\eps}\dY{p}{q}$.
\end{proof}

\begin{theorem}
    \thmlab{d-dim}%
    Let $\epsR,\eps \in \pth{0,1}$ be fixed parameters and
    $\PS\subseteq [0,1)^d$ be a set of $n$ points. The graph $\Graph$,
    constructed in \secref{higher_dim_constr}, is a $\epsR$-reliable
    $\pth{1+\eps}$-spanner of $\PS$ in expectation and has size
    \begin{math}
        \Of\bigl( \constA \, n \log\!\log^2\!n \log\!\log\!\log n \bigr),
    \end{math}
    where $\Of$ hides constant that depends on the dimension $d$, and
    \begin{math}
        \constA%
        =%
        \Of(\eps^{-2d}\epsR^{-1} \log^3 \eps^{-1}\log \epsR^{-1}).
    \end{math}
\end{theorem}
\begin{proof}
    The size of the construction is proved in \lemref{size-d-dim}. Let
    $\BSet_N$ be the set defined above. By \lemref{bad-points-d-dim},
    the expected size of $\BSet_N$ is at most
    $\pth{1+\epsR}\cardin{\BSet}$. By \lemref{1+eps-path}, for any two
    points $p,q\in \PS \setminus \BSet_N$, there is a $(1+\eps)$-path
    between $p$ and $q$ in the graph $\Graph\setminus\BSet$. Thus, we
    have $\Ex{\lossY{\Graph}{\BSet}} \leq \epsR$.
\end{proof}

\subsection{Probabilistic bound}
\seclab{d-dim-analysis-prob}

The same construction, as we used in \secref{higher_dim_constr}, can be
applied to construct spanners with near linear edges that are reliable
with probability $1-\epsP$. The idea is to use the probabilistic version of the one-dimensional construction with parameters $\epsP' = \frac{\epsP}{MN}$ and $\epsR' = \frac{\epsR}{3MN}$. Then, similarly to \lemref{bad-points-d-dim}, it is not hard to show that $\cardin{\BSet_N} \leq \pth{1+\epsR} \cardin{\BSet}$ holds with probability $1-\epsP$.

\begin{lemma}
    \lemlab{bad-points-d-dim-prob} Let $\BSet_N$ be the set defined
    in \secref{higher_dim_anal}. The probability that the size of $\BSet_N$ is larger than
    $\pth{1+\epsR}\cdot \cardin{\BSet}$ is at most $\epsP$.
\end{lemma}
\begin{proof}
    By \lemref{bad-points-prob}, and since all attacks are oblivious,
    we have
    $$\Prob{\cardin{\BSet_i^\order} > (1+\epsR') \cardin{\BSet_{i-1}}}
    \leq \epsP'$$ for all $\order \in \ordAll$ and $i\geq
    1$. Therefore,
    \begin{align*}
      \Prob{\cardin{\BSet_i} > \pth{1+M\epsR'}\cardin{\BSet_{i-1}}}
      &=
        \Prob{\cardin{\BSet_i \setminus \BSet_{i-1}} > M\epsR'\cardin{\BSet_{i-1}}} \\
        &\leq
        \Prob{\cup_{\order \in \ordAll} \cardin{\BSet_i^\order \setminus \BSet_{i-1}} > \epsR'\cardin{\BSet_{i-1}}} \\
      &\leq
        \sum_{\order \in \ordAll} \Prob{\cardin{\BSet_i^\order \setminus \BSet_{i-1}} > \epsR'\cardin{\BSet_{i-1}}} \\
        &=
        \sum_{\order \in \ordAll} \Prob{\cardin{\BSet_i^\order} > (1+\epsR')\cardin{\BSet_{i-1}}}
        \leq
        M\epsP'
    \end{align*}
    holds for $i=1,\dots,N$. Since
    $\pth{1+\frac{\epsR}{3N}}^N \leq \pth{e^{{\epsR}/{3N}}}^N \leq
    1+\epsR$, we get
    \begin{align*}
      \Prob{\cardin{\BSet_N} > (1+\epsR) \cardin{\BSet}\bigr.}
      &\leq
        \Prob{\cardin{\BSet_N} >
        \Bigl({1+\frac{\epsR}{3N}}\Bigr)^N \cardin{\BSet}} \\
      &\leq
        \Prob{\bigcup_{i=1}^N \cardin{\BSet_i} >
        \pth{1+\frac{\epsR}{3N}} \cardin{\BSet_{i-1}}} \\
      &\leq
        \sum_{i=1}^N \Prob{\cardin{\BSet_i} >
        \pth{1+\frac{\epsR}{3N}} \cardin{\BSet_{i-1}}} \\
      &=
        \sum_{i=1}^N \Prob{\cardin{\BSet_i} > \pth{1+M\epsR'}
        \cardin{\BSet_{i-1}}}
        \leq
        NM\epsP'
        =
        \epsP.
        \qedhere
    \end{align*}
\end{proof}

Therefore, using the same argument as for \thmref{d-dim}, we obtain the following result, which gives a slight improvement in the constants, compared to the trivial multiplicative factor $\Of\pth{\log\epsP^{-1}}$ by simply repeating the construction of \secref{higher_dim_constr}.

\begin{theorem}
    \thmlab{d-dim-prob} Let $\epsR,\eps,\epsP \in \pth{0,1}$ be fixed
    parameters and $\PS\subseteq [0,1)^d$ be a set of $n$ points. The
    graph described above is a
    $\epsR$-reliable $\pth{1+\eps}$-spanner of $\PS$ with probability
    $1-\epsP$. %
    The size of the construction is
    \begin{math}
        \Of\bigl( \constA \, n \log\!\log^2\!n \log\!\log\!\log n \bigr),
    \end{math}
    where $\Of$ hides constant that depends on the dimension $d$, and
    \begin{math}
        \constA%
        =%
        \Of\bigl(\eps^{-2d}\epsR^{-1} \log^3 \eps^{-1} (\log \epsR^{-1}
        + \log \epsP^{-1})\bigr).
    \end{math}
\end{theorem}

\section{Conclusions}
\seclab{conclusion}

Reliable spanners require $\Omega (n \log n)$ edges. In this paper, we
showed that fewer edges are sufficient, if the spanner only has to be
reliable against oblivious attacks (in expectation or with a certain
probability). Our new construction avoids the use of expanders, and as
a result has much smaller constants than previous constructions,
making it potentially practical. The number of edges in the new
spanner is significantly smaller -- it is linear in one dimension, and
roughly $\Of(n \log\!\log^2\!n )$ in higher dimensions.  An open problem
is whether these $\log\!\log$-factors in higher dimensions can be
avoided. Furthermore, similar results for reliable spanners for general
metrics would be of interest.

\BibTexMode{%
   \bibliographystyle{alpha}%
   \bibliography{qreliable}%
}%
\BibLatexMode{\printbibliography}

\end{document}